\documentclass[lettersize,journal]{IEEEtran}

\usepackage{amsmath,amssymb,amsthm,amsfonts,mathrsfs,dsfont}
\usepackage{graphicx}
\usepackage{subcaption}
\usepackage{float}
\usepackage{verbatim}
\usepackage{epstopdf}
\usepackage{color}
\usepackage{cite}
\usepackage{cancel}
\usepackage[shortlabels]{enumitem}
\usepackage[breaklinks=true,letterpaper=true,colorlinks=false,bookmarks=false]{hyperref}
\usepackage{algorithm}
\usepackage{algpseudocode}
\usepackage{arydshln}
\usepackage[official]{eurosym}
\usepackage{comment}

\usepackage[colorinlistoftodos,prependcaption,textsize=tiny]{todonotes}

\usepackage{tikz}
\usetikzlibrary{chains,shapes.misc,positioning}

\usepackage{pgfplots}
\pgfplotsset{compat=newest}

\newtheorem{theorem}{Theorem}

\newtheorem{remark}{Remark}

\newtheorem{lemma}{Lemma}

\pgfplotscreateplotcyclelist{nano mark style}{%
dashed, every mark/.append style={solid, fill=gray!30}, mark=*, line width=1.1pt\\%
dashed, every mark/.append style={solid, fill=gray!30}, mark=square*, line width=1.1pt\\%
dashed, every mark/.append style={solid, fill=gray!30}, mark=triangle*, mark size=2.5pt, line width=1.1pt\\%
dashed, every mark/.append style={solid, fill=gray!30}, mark=diamond*, mark size=2.5pt, line width=1.1pt\\%
dashed, every mark/.append style={solid, fill=gray!30}, mark=x, mark size=3pt, line width=1.1pt\\%
}

\pgfplotscreateplotcyclelist{nano line style}{%
color=black, line width=2pt\\%
densely dashed, color=black, line width=2pt\\%
densely dotted, color=black, line width=2pt\\%
loosely dashed, color=black, line width=2pt\\%
}

\begin{document}

\title{Differential Privacy over Hamming Codes}


\author{Borzoo Rassouli$^1$ and Morteza Varasteh$^2$\\
\small{$^1$Nokia Bell Labs, Stuttgart, Germany}\\
\small{$^2$University of Essex, Colchester, UK}\\
{\tt\small borzoo.rassouli@nokia-bell-labs.com},
{\tt\small m.varasteh@essex.ac.uk}
}

\maketitle
\begin{abstract} We consider the transmission of the outputs of counting queries over a binary symmetric channel (BSC), where Hamming codes are employed as the channel encoder. Since the channel is inherently noisy, this transmission already provides a degree of privacy protection “for free”, albeit at the cost of reduced utility in the form of decoding errors.
A natural question is whether this privacy can be further improved (i) without any additional real-time obfuscation of the data, such as injecting artificial noise prior to transmission, and (ii) without increasing the end-to-end error probability. In this work, we answer this question in the affirmative by deriving an optimal codeword arrangement that strictly improves differential privacy guarantees while incurring no real-time computational overhead and no degradation in utility.
\end{abstract}
\begin{IEEEkeywords}
Differential privacy, privacy for free, Hamming codes, binary symmetric channel, counting queries.
\end{IEEEkeywords}
\section{introduction}

Differential privacy (DP) has emerged as one of the most widely adopted privacy-enhancing technologies due to its rigorous and quantifiable privacy guarantees\cite{DP2014}. It provides a mathematical framework for limiting the information that can be inferred about any individual contributing to a dataset while enabling useful statistical analysis and machine learning tasks. 

Classical DP mechanisms achieve privacy by deliberately perturbing the released data, typically through the injection of randomized noise. While this approach provides provable privacy guarantees, it introduces an inherent \emph{utility-privacy tradeoff}, whereby stronger privacy protection generally requires larger distortions of the underlying data, which is not only restricted to DP \cite{Zamani}. 

Roughly speaking, noise and randomness are desirable from a DP perspective, as they help obscure sensitive information. In contrast, from a communications perspective, noise is generally undesirable, and considerable effort is devoted to mitigating its effects through sophisticated coding and signal-processing techniques. However, in the context of private communication, it is natural to ask whether the inherent channel noise can be exploited, rather than combated, as a source of privacy enhancement. 

This perspective has recently inspired a growing body of research on so-called \emph{privacy-for-free} mechanisms. Representative examples include dataset condensation techniques, posterior sampling in Bayesian learning, approximate homomorphic encryption schemes, overparameterized learning, communication-efficient federated learning via sketching, and wireless learning systems that exploit channel noise or hardware impairments to provide privacy guarantees without dedicated privacy-preserving operations \cite{dong2022privacy, wang2015privacy, Ogilvie23, Bombari,Smith19, Simeone21, Simeone_p4f22, ADC_p4f23}. The present letter falls within this emerging line of work.
Turning to a different aspect, a fundamental class of statistical releases in differential privacy is that of counting queries, whose output is an integer-valued statistic, such as the number of users possessing a given attribute \cite{counting}. This paper investigates the enhancement of DP, within the privacy-for-free paradigm, for transmitting the outputs of such queries over a noisy channel (BSC) with Hamming codes as the encoder. 

The analysis developed in this letter relies on two key structural properties. First, the codewords must admit an ordering in which consecutive codewords are separated by the minimum code distance. Equivalently, the minimum-distance graph of the code must contain a Hamiltonian path. This, in turn, enables the systematic assignment of neighboring query outputs to codewords that are maximally confusable from a differential privacy perspective. Second, the resulting privacy guarantees depend on the transition probabilities induced by \emph{maximum likelihood} (ML) decoding over the BSC. To obtain an exact characterization of privacy loss, and hence the claim of optimality, these transition probabilities must admit a tractable analytical representation. Hamming codes constitute a particularly attractive setting because both properties hold simultaneously.

In summary, this work makes the following contributions. First, we characterize the minimum achievable privacy loss for the transmission of counting-query outputs over a BSC using Hamming codes. Second, we present an explicit construction, based on a weight-three generator representation and Gray encoding, that attains this optimal privacy loss. Importantly, the proposed scheme preserves the decoding error probability of the underlying communication system and incurs no additional online computational overhead.

\section{Preliminaries}
\subsection{Differential Privacy}
Two datasets $D_1$ and $D_2$ are neighbors (or adjacent) if they differ in exactly one data entry (i.e., the data of a single individual). Let $\epsilon \geq 0$. A (possibly randomized) algorithm or mechanism $\mathcal{A}$ is said to be $\epsilon$-differentially private if
\begin{equation}
    \Pr\{\mathcal{A}(D_1)\in \mathds{S}\}
    \leq e^\epsilon \Pr\{\mathcal{A}(D_2)\in \mathds{S}\},
\end{equation}
for all subsets $\mathds{S}$ of the range of $\mathcal{A}$ and all adjacent datasets $D_1$ and $D_2$.

The smaller $\epsilon$ is, the harder it becomes to distinguish whether $D_1$ or $D_2$ generated the output of the algorithm. This, in turn, makes it less likely to infer the presence or absence of any individual, since $D_1$ and $D_2$ differ in only one entry. Therefore, decreasing $\epsilon$ increases the level of privacy protection provided by the mechanism. As a result, $\epsilon$ is often referred to as the \textit{privacy loss} parameter.


\subsection{Hamming Codes}
Hamming codes (\cite{Hamming}) are a class of binary linear error-correcting codes that can correct single-bit errors. For each integer $q \geq 2$, there exists a Hamming code with block length $n = 2^q - 1$ and message length $k = 2^q - q - 1$ denoted by Hamming$(n,k)$. The parity-check matrix $\mathbf{H}$ of a Hamming code is constructed by listing all nonzero binary vectors of length $q$ as its columns, resulting in a $q \times n$ matrix of rank $q$. The null space of $\mathbf{H}$, which is a linear subspace of dimension $k\ (=n-q)$ in an $n$-dimensional vector space, is taken to be the set of codewords, i.e., a set of $2^k$ binary codewords of length $n$ denoted by $\mathcal{C}$.

Except for the all-zero codeword, the minimum number of nonzero entries in any codeword—called the \textit{minimum weight} of the code—is 3. This follows from the fact that the sum (over $\mathbb{F}_2$) of any two columns of $\mathbf{H}$ is again a column of $\mathbf{H}$. Since the code is linear, the difference between any two codewords is also a codeword. Therefore, the minimum number of positions in which any two codewords differ—called the \textit{minimum distance} of the code—is also 3. As a result, if a codeword $\mathbf{c}$ is corrupted in at most one position, it remains closer (in Hamming distance) to $\mathbf{c}$ than to any other codeword, and the ML decoder can uniquely decode it.

\begin{remark}
For counting queries, the output is an integer-valued statistic, and the outputs corresponding to neighboring datasets differ by at most one. Therefore, adjacency in the dataset domain induces adjacency in the integer-valued output space. When such outputs are transmitted over a communication channel, DP requires an appropriate mapping of outputs to codewords that maximizes their indistinguishability at the receiver, subject to a privacy parameter $\epsilon$.
\end{remark}


\section{Main results}
In this section, we assume that the outputs of a counting query are transmitted over a binary symmetric channel with crossover probability $p\in(0,\tfrac{1}{2})$, denoted BSC$(p)$, using Hamming codes as the channel encoder. Although the channel noise is undesirable from a reliability and utility perspective, it inherently induces a certain level of privacy protection.

In what follows, we propose a method to enhance this level of privacy without introducing additional privatizing steps (such as explicit noise injection prior to transmission) and without increasing the message error probability. In summary, the proposed privacy enhancement incurs no real-time computational overhead and no degradation in utility.

In the context of differential privacy, the objective of a privacy-preserving mapping is to render adjacent datasets sufficiently indistinguishable. For transmission of the outputs of a counting query using Hamming codes, a natural heuristic for improving privacy (i.e., reducing privacy loss) is to arrange the codewords such that adjacent datasets are mapped to codewords with minimum possible Hamming distance.

However, this intuition raises two questions: (i) whether such an arrangement is feasible, and (ii) whether it is optimal in terms of privacy loss.

In this section, we answer these two questions in the affirmative. Moreover, since such an ordering corresponds only to a permutation of message labels, it does not affect the message error probability. 

Let $d_H(\cdot,\cdot)$ denote the Hamming distance between two binary codewords, i.e., the number of positions in which they differ. The following three lemmas are needed in the sequel for the analysis of the privacy loss.

\begin{lemma}\label{lem1}
In the transmission of Hamming$(n,k)$ codewords over a BSC$(p)$, the probability of decoding $\mathbf{c}'$ when $\mathbf{c}$ is transmitted is given by
\begin{equation}\label{transp}
    \Pr\{\mathbf{c}'|\mathbf{c}\}
    = p^d(1-p)^{n-d}\left(1+\frac{1-p}{p}d+\frac{p}{1-p}(n-d)\right),
\end{equation}
for all $\mathbf{c},\mathbf{c}'\in\mathcal{C}$, where $d \triangleq d_H(\mathbf{c},\mathbf{c}')$.
\end{lemma}

\begin{proof}
The proof is provided in Appendix \ref{AppA}.
\end{proof}

For fixed $n,p$, let $f(d)$ denote the transition probability in \eqref{transp} as a function of $d$.

\begin{lemma}\label{lem2}
For nonzero integer $m$, the ratio $\frac{f(d+m)}{f(d)}$ is strictly decreasing (increasing) in $d$ if $m$ is positive (negative). Furthermore, $f(d)$ is strictly decreasing in $d$.
\end{lemma}

\begin{proof}
The proof is provided in Appendix \ref{AppB}.
\end{proof}
\begin{lemma}\label{lem3}
    In Hamming$(n,k)$ code, there exists a set of $k$ codewords of weight $3$ that form a basis for the null space of the parity-check matrix $\mathbf{H}$. In other words, every codeword can be expressed as a linear combination of these $k$ basis codewords. 
\end{lemma}
\begin{proof}
The proof is provided in Appendix \ref{AppC}.

\end{proof}

We are now ready to establish the main result of this section, stated in the following theorem.
\begin{theorem}
    The minimum privacy loss in transmission of the outputs of a counting query over a BSC($p$) via Hamming$(n,k)$ codes is 
    \begin{equation}\label{opteps}
        \epsilon^*=3\ln\frac{1-p}{p}-\ln\frac{(1-p)n+p}{(1-p)n+p-3\frac{1-2p}{1-p}},
    \end{equation}
    and it is achieved if and only if the codewords are arranged such that $d_H(\mathbf{c}_i,\mathbf{c}_{i+1})=3,\ \forall i\in[0:2^k-2]$. 
\end{theorem}
\begin{proof}
The proof is completed by establishing the following steps:
\begin{enumerate}

\item For any arrangement of codewords in which there exists at least one pair of adjacent codewords whose Hamming distance exceeds $3$, the privacy loss is strictly greater than $\epsilon^*$.

\item For any arrangement of codewords in which every pair of adjacent codewords has Hamming distance equal to $3$, the privacy loss is $\epsilon^*$.

\item An arrangement of codewords as in the previous step can be constructed.

\end{enumerate}

Let $\mathcal{M} \triangleq [0:2^k-1]$ denote the message set, and let $p_{j,i} \triangleq \Pr\{\mathbf{c}_j \mid \mathbf{c}_i\},\ \forall i,j \in \mathcal{M}$,
denote the probability that the decoder outputs $\mathbf{c}_j$ when $\mathbf{c}_i$ is transmitted.

From the definition of differential privacy, it follows that the privacy loss $\epsilon$ in the transmission is given by
\begin{equation}\label{privacyloss}
\epsilon
= \ln \max_{\substack{i,j,k \in \mathcal{M} \ , \ |i-k|=1}}
\frac{p_{j,i}}{p_{j,k}},
\end{equation}
where $|i-k|=1$ corresponds to neighboring query outputs. In other words, the privacy loss is the natural logarithm of the maximum ratio between the probabilities of producing the same decoded output when two neighboring codewords are transmitted.

For simplicity, define
$d_{j,i} \triangleq d_H(\mathbf{c}_j,\mathbf{c}_i).
$ Hence, we can write $p_{j,i} = f(d_{j,i}), \forall i,j \in \mathcal{M}$,
where $f(\cdot)$ is the error probability defined prior to Lemma~\ref{lem2}. 

Step 1 is proved as follows. Consider an arbitrary arrangement in which there exist two adjacent codewords $\mathbf{c}_i$ and $\mathbf{c}_{i+1}$ such that $d_{i,i+1}\triangleq m$ with $m>3$. Let $\mathbf{c}_j$ denote the modulo-two sum of $\mathbf{c}_i$ and the all-one codeword, and from the linearity of the code, $\mathbf{c}_j$ is also a codeword. Obviously, $d_{j,i}=n$ and $d_{j,i+1}=n-m$, and from the definition of $\epsilon$, we have
 \begin{align}
     \epsilon&\geq\ln\frac{p_{j,i+1}}{p_{j,i}}\nonumber\\
     &=\ln\frac{f(d_{j,i+1})}{f(d_{j,i})}\nonumber\\
     &=\ln\frac{f(n-m)}{f(n)}\nonumber\\
     &>\ln\frac{f(n-3)}{f(n)}\label{mono}\\
     &=\epsilon^*,\label{mono1}
 \end{align}
where in (\ref{mono}), we used the fact that $f(d)$ is strictly decreasing in $d$, as established in Lemma~\ref{lem2}, and $m>3$. The last equality in (\ref{mono1}) follows by evaluating (\ref{mono}). Therefore, the resulting privacy loss is strictly greater than $\epsilon^*$, implying that such an arrangement is not optimal with respect to minimizing the privacy loss. 

For step 2, we proceed as follows. Assume that we have an arrangement in which all the consecutive codewords have the minimum distance of three. Therefore, the privacy loss can be written as
\begin{align}
\epsilon
&= \ln \max_{\substack{i,j,k \in \mathcal{M} \ |i-k|=1}} \frac{p_{j,i}}{p_{j,k}} \nonumber\\
&= \ln \max_{j \in \mathcal{M}} \max_{\substack{i,k \in \mathcal{M} \ |i-k|=1}} \frac{p_{j,i}}{p_{j,k}} \nonumber\\
&= \ln \max_{j \in \mathcal{M}} \max_{\substack{i,k \in \mathcal{M} \ |i-k|=1}} \frac{f(d_{j,i})}{f(d_{j,k})}\nonumber\\
&= \ln \max_{j \in \mathcal{M}} \max_{\substack{i \in \mathcal{M} \setminus {2^k-1}}}
\max \left\{
\frac{f(d_{j,i})}{f(d_{j,i+1})},
\frac{f(d_{j,i+1})}{f(d_{j,i})}
\right\} \label{num4}\\
&= \ln \max_{j \in \mathcal{M}} \max_{\substack{i \in \mathcal{M} \setminus {2^k-1}}}
\frac{f\big(\min\{d_{j,i}, d_{j,i+1}\}\big)}
{f\big(\max\{d_{j,i}, d_{j,i+1}\}\big)}. \label{eq4}
\end{align}
\begin{align}
&\leq \ln \max_{j \in \mathcal{M}} \max_{\substack{i \in \mathcal{M} \setminus {2^k-1}}}
\frac{f\big(\max\{d_{j,i}, d_{j,i+1}\}-3\big)}
{f\big(\max\{d_{j,i}, d_{j,i+1}\}\big)} \label{eq5}\\
&\leq \ln \max_{j \in \mathcal{M}} \max_{\substack{i \in \mathcal{M} \setminus {2^k-1}}}
\frac{f(n-3)}{f(n)} \label{eq6}\\
&= \ln \frac{f(n-3)}{f(n)}\nonumber \\
&= \epsilon^*,\nonumber
\end{align}
where (\ref{num4}) accounts for the fact that $k$ is either $i-1$ or $i+1$, and (\ref{eq4}) follows from the fact that $f(d)$ is strictly decreasing in $d$ by Lemma~\ref{lem2}. Moreover, (\ref{eq5}) follows from the monotonicity of $f(\cdot)$ together with the triangle inequality
\begin{align*}
    d_{j,i}
    &\leq d_{j,k} + d_{k,i}
    = d_{j,k} + 3, \quad \forall\, i,j,k \in \mathcal{M}: |i-k|=1.
\end{align*}
Therefore, we obtain $\min\{d_{j,i}, d_{j,i+1}\}
\geq \max\{d_{j,i}, d_{j,i+1}\} - 3.$

Lemma~\ref{lem2} results in (\ref{eq6}), since the ratio $\frac{f(d-3)}{f(d)}$ is strictly increasing in $d$ and therefore attains its maximum at $d=n$. 

 Thus far, we have shown that for any arrangement in which the distance of any two adjacent codewords is 3, we have $\epsilon\leq\epsilon^*$. To show that $\epsilon=\epsilon^*$, it suffices to show $\epsilon\geq\epsilon^*$. However, this can be readily derived from the proof of step 1. More specifically, put $m=3$ and eliminate the strict inequality in (\ref{mono}), and we get
 \begin{align}
     \epsilon&\geq\ln\frac{p_{j,i+1}}{p_{j,i}}\nonumber\\
     &=\ln\frac{f(n-3)}{f(n)}\nonumber\\
     &=\epsilon^*.\nonumber
 \end{align}
This completes the proof of the converse.

For Step 3, we may proceed in two ways. The first and shorter approach is to invoke Lemma~\ref{lem3}. Specifically, Lemma~\ref{lem3} implies that the minimum-distance graph of the code is a connected finite Cayley graph and therefore admits a Hamiltonian path. The existence of such a path is equivalent to the existence of an arrangement of the codewords in which every pair of consecutive codewords differs in exactly three positions. This completes the proof.

Rather than relying on this indirect graph-theoretic argument, however, we adopt a constructive approach and explicitly build such an arrangement. To this end, it suffices to identify the basis of $k$ weight-$3$ codewords described in Lemma~\ref{lem3}.

Arrange the columns of the parity-check matrix $\mathbf{H}$ in decreasing order of their decimal representations except for the identity matrix $\mathbf{I}$ that comes at the end. For example, for the Hamming code with $(n,k)=(15,11)$, we obtain
\begin{equation*}
\mathbf{H}=
\begin{bmatrix}
1 & 1 & 1 & 1 & 1 & 1 & 1 & 0 & 0 & 0 & 0 & 1 & 0 & 0 & 0 \\
1 & 1 & 1 & 1 & 0 & 0 & 0 & 1 & 1 & 1 & 0 & 0 & 1 & 0 & 0 \\
1 & 1 & 0 & 0 & 1 & 1 & 0 & 1 & 1 & 0 & 1 & 0 & 0 & 1 & 0 \\
1 & 0 & 1 & 0 & 1 & 0 & 1 & 1 & 0 & 1 & 1 & 0 & 0 & 0 & 1
\end{bmatrix}.
\end{equation*}
We construct \(k\) linearly independent codewords of weight \(3\) as follows. For each \(i\in[1:k]\), let \(\mathbf{c}_i\) denote a codeword having a \(1\) in position \(i\) and \(0\)s in positions \([1:i-1]\). This triangular structure guarantees linear independence, since \(\mathbf{c}_i\) cannot be expressed as a linear combination of any subset of the remaining \(k-1\) codewords.

It remains to determine two additional positions in which to place \(1\)s so that \(\mathbf{c}_i\) has weight \(3\) and satisfies $\mathbf{H}\mathbf{c}_i^T=\mathbf{0}$.
Equivalently, the sum of the two columns of \(\mathbf{H}\) corresponding to these additional positions must equal column \(i\). There are several ways to select such positions; one of them is described below.

Let \(\pi:[1:n]\to[1:n]\) be the one-to-one mapping that assigns to each column index of \(\mathbf{H}\) the decimal value of the corresponding binary column vector. In the example under consideration, the eleventh column is \([0\ 0\ 1\ 1]^T\), whose decimal value is \(3\), and therefore \(\pi(11)=3\). The remaining values of \(\pi\) are defined analogously. 

For each \(i\in[1:k]\), let \(R(i)\) denote the position of the most significant \(1\) in column \(i\). More precisely,   $R(i)\triangleq \max_{2^{j-1}\leq \pi(i)}j.$
For example, we have $R(11)=2$, and so on. 

To determine the remaining two positions in codeword $i$ that are to be assigned the value $1$, we decompose column \(i\) into the sum of its most significant bit and the vector formed by its remaining bits.  For example, column 3, namely $[1\ 1\ 0\ 1]^T$ can be written as $[1\ 0\ 0\ 0]^T+[0\ 1\ 0\ 1]^T$, which correspond to columns 12 and 10. Hence, codeword 3 has 1's in positions 3, 12 and 10, and zeros elsewhere. This construction can be described algorithmically as follows. The index corresponding to the most significant bit is $n+1-R(i).$ In the above example, $16-R(3)=12$.
The index of the second additional column is obtained from its decimal value. Specifically, the decimal value of column \(i\) is \(\pi(i)\), while the decimal value of its most significant bit is \(2^{R(i)-1}\). Therefore, the decimal value of the remaining part is $\pi(i)-2^{R(i)-1}$,
and the corresponding column index is $\pi^{-1}\left(\pi(i)-2^{R(i)-1}\right)$. For the example above, \(\pi(3)=13\) and \(2^{R(3)-1}=8\), which yields $\pi^{-1}(13-8)=\pi^{-1}(5)=10$, as claimed.

In general, for each \(i\in[1:k]\), define codeword \(\mathbf{c}_i\) by placing \(1\)s in positions
$i, n-R(i)+1, \text{ and } \pi^{-1}\left(\pi(i)-2^{R(i)-1}\right),
$
and \(0\)'s in all remaining positions.
By construction, each \(\mathbf{c}_i\) has weight \(3\), satisfies \(\mathbf{H}\mathbf{c}_i^T=\mathbf{0}\), and the collection \(\{\mathbf{c}_1,\ldots,\mathbf{c}_k\}\) is linearly independent. Consequently, these \(k\) codewords form a basis of the code and may be taken as the rows of a generator matrix $\mathbf{G}_{k\times n}$.


Afterwards, we first apply Gray encoding to the message \(\mathbf{x}\) to obtain \(\hat{\mathbf{x}}\). The released codeword is then given by \(\mathbf{c}(\mathbf{x})=\hat{\mathbf{x}}\mathbf{G}.\)

Since any two adjacent Gray-encoded messages differ in exactly one bit position, their corresponding codewords differ by the addition (equivalently, subtraction over \(\mathbb{F}_2\)) of a single row of \(\mathbf{G}\). Because every row of \(\mathbf{G}\) has weight \(3\), the Hamming distance between the corresponding codewords is exactly \(3\). Therefore, adjacent messages are mapped to codewords that are separated by the minimum possible distance.

The overall procedure is summarized in Algorithm~\ref{alg2}. The algorithm is executed offline only once, and the resulting mapping is subsequently used for real-time operation.

\begin{algorithm}
\caption{Constructing the arrangement for Hamming$(n,k)$}\label{alg2}

\begin{algorithmic}[1]
\State Write $\mathbf{H}=[\mathbf{A}|\mathbf{I}]$, and arrange the columns of A in decreasing order of their decimal values.
\State Define $\pi(i)\triangleq \textnormal{bin2dec(column $i$ of $\mathbf{H}$)},\ i\in[1:n]$, and $R(i)\triangleq \max_{2^{j-1}\leq \pi(i)}j,\ i\in[1:k]$.
\State Construct $\mathbf{G}_{k\times n}$ with row $i$ having ones in positions $i,n-R({i})+1$, and $\pi^{-1}\left(\pi(i)-2^{R({i})-1}\right)$ and zeros elsewhere.
\State Output $\hat{\mathbf{x}}\mathbf{G}$ for the Gray encoded message $\hat{\mathbf{x}}$.

\end{algorithmic}
\end{algorithm} 
\end{proof}
It can be verified that in an arbitrary arrangement of codewords, the privacy loss can be as large as
\begin{equation}\label{worstcase}
   \epsilon_{\textnormal{max}}\triangleq (n-1)\ln\frac{1-p}{p}-\ln\frac{n(1-p)+p}{np+(1-p)}, 
\end{equation}
which occurs when the distance of two neighboring codewords is $n$. Comparing with $\epsilon^*$ in \eqref{opteps}, we have
\begin{equation*}
    \frac{\epsilon_{\textnormal{max}}}{\epsilon^*}=\frac{n-2}{3}+C_p+O(\frac{1}{n}),
\end{equation*}
where $C_p\triangleq\frac{1-2p}{3(1-p)\ln(\frac{1-p}{p})}$ is a constant.
\section{Conclusion}
This letter investigated the differential privacy properties of transmitting counting-query outputs over a BSC using Hamming codes. Rather than introducing additional privatization mechanisms, we exploited the inherent randomness already present in the communication channel and studied how the assignment of query outputs to codewords affects privacy. We derived the minimum achievable privacy loss for Hamming-coded transmission and proved that it is attained when adjacent query outputs are mapped to codewords at the minimum Hamming distance. An explicit construction based on a weight-three generator representation and Gray encoding was presented to realize this optimal arrangement. Importantly, the proposed method preserves the decoding error probability and incurs no additional online computational overhead. These results demonstrate that meaningful privacy improvements can be obtained through codeword organization alone, highlighting a new connection between channel coding and differential privacy. \footnote{This work is partially part of the Unified Networking Experience (UNEXT) framework presented in \cite{UNEXT}.}

\appendices
\section{Proof of lemma \ref{lem1}}\label{AppA}
Under ML decoding, Hamming codes correct any single-bit error. Therefore, a codeword $\mathbf{c}'$ is decoded at the receiver if and only if the received vector $\mathbf{r}$ satisfies $d_H(\mathbf{r},\mathbf{c}') \leq 1$. This condition is equivalent to either $d_H(\mathbf{r},\mathbf{c}') = 0$ or $d_H(\mathbf{r},\mathbf{c}') = 1$.  Since $\mathbf{c}$ and $\mathbf{c}'$ differ in $d$ positions, we have $d_H(\mathbf{r},\mathbf{c}') = 0$ if and only if all $d$ differing bits are flipped and the remaining $n-d$ matching bits remain unchanged during the transmission of $\mathbf{c}$. This occurs with probability $p^d(1-p)^{n-d}$.

For the case $d_H(\mathbf{r},\mathbf{c}') = 1$, there are two possibilities: (i) exactly $(d-1)$ of the $d$ differing bits are flipped and all $(n-d)$ matching bits remain unchanged, which occurs with probability $d\,p^{d-1}(1-p)^{n-d+1}$; or (ii) all $d$ differing bits are flipped and exactly one of the $(n-d)$ matching bits is also flipped, while the remaining $(n-d-1)$ matching bits remain unchanged, which occurs with probability $(n-d)p^{d+1}(1-p)^{n-d-1}$.

Hence, given that $\mathbf{c}$ was transmitted, the probability of the event $d_H(\mathbf{r},\mathbf{c}') \leq 1$ is
\begin{align*}
    \Pr\{\mathbf{c}'|\mathbf{c}\}
    &= p^d(1-p)^{n-d}
    + d\,p^{d-1}(1-p)^{n-d+1}\nonumber\\
    &\ \ \ + (n-d)p^{d+1}(1-p)^{n-d-1} \\
    &= p^d(1-p)^{n-d}\left(1+\frac{1-p}{p}d+\frac{p}{1-p}(n-d)\right).
\end{align*}
\section{Proof of lemma \ref{lem2}}\label{AppB}
For the first claim, we extend the domain of $f(\cdot)$ to the real numbers and show that $\frac{\partial}{\partial d}\frac{f(d+m)}{f(d)}$ and $m$ have opposite signs.

For simplicity, define $a \triangleq \frac{1-2p}{p(1-p)}$ and $b \triangleq 1+\frac{np}{1-p}$ as two non-negative constants. Then $f(d) = p^d(1-p)^{n-d}(ad+b)$, and we obtain
\begin{align}
   m\frac{\partial}{\partial d}\frac{f(d+m)}{f(d)}
   &= m\left(\frac{p}{1-p}\right)^m \frac{\partial}{\partial d}\frac{a(d+m)+b}{ad+b} \nonumber\\
   &= -\left(\frac{p}{1-p}\right)^m (\frac{ma}{ad+b})^2 \nonumber\\
   &< 0.
\end{align}

The second claim is proved by showing that $f(d+1) < f(d)$, or equivalently, that $\frac{f(d+1)}{f(d)} < 1$. From the first claim, $\frac{f(d+1)}{f(d)}$ is strictly decreasing in $d$. Therefore, it suffices to show that $\frac{f(1)}{f(0)} < 1$. We have

\begin{align}
    \frac{f(1)}{f(0)}
    &= \frac{p}{1-p}\left(1+\frac{a}{b}\right) \nonumber\\
    &< \frac{p}{1-p}\left(1+\frac{a}{1+\frac{p}{1-p}}\right) \label{eq2}\\
    &= 1,
\end{align}
where \eqref{eq2} follows from the inequality $b \triangleq 1+\frac{np}{1-p}>1+\frac{p}{1-p}$ since $n>1$ (indeed, $n\geq 3$ for Hamming codes). 
\section{Proof of lemma \ref{lem3}}\label{AppC}
In Hamming$(n,k)$ code, a vector $\mathbf{c}$ is a codeword if and only if $\mathbf{H}\mathbf{c}^T=\mathbf{0}$. Let $\mathcal{W}_3$ denote the set of all codewords of weight $3$. Since the sum of any two distinct columns of $\mathbf{H}$ is also a column of $\mathbf{H}$, each pair of columns determines a unique weight-$3$ codeword. Consequently, the total number of weight-$3$ codewords is $|\mathcal{W}_3|=\frac{1}{3}\binom{n}{2}$, where the factor $\frac{1}{3}$ accounts for the fact that each weight-$3$ codeword is counted three times, corresponding to the three possible pairs among its three nonzero positions.

In what follows, we show that every codeword can be expressed as a linear combination of codewords in $\mathcal{W}_3$. It follows that $\mathcal{W}_3$ spans the null space of $\mathbf{H}$. Since the latter, namely the code itself, has dimension $k$, we conclude that there exists a subset of $k$ codewords in $\mathcal{W}_3$ that forms a basis of the code.\footnote{Note that $k<\frac{1}{3}\binom{n}{2}$, which follows by substituting $n=2^q-1$ and $k=2^q-q-1$ for $q\ge 2$.}

Let $\mathbf{c}$ be an arbitrary codeword of weight $m>3$, i.e., $\mathbf{c}$ contains $m$ ones and $n-m$ zeros. Clearly, $\mathbf{H}\mathbf{c}^T=\mathbf{0}$. Let $i_1,i_2,\ldots,i_m$ denote the indices of ones in $\mathbf{c}$.
Choose two distinct indices, say $i_1$ and $i_2$. The sum of the corresponding columns of $\mathbf{H}$ equals another column of $\mathbf{H}$, whose index we denote by $i^*$. Construct the weight-$3$ codeword $\mathbf{c}'$ with ones in positions $i_1$, $i_2$, and $i^*$. Define $\mathbf{c}_{\textnormal{new}} \triangleq \mathbf{c} \oplus \mathbf{c}'$,
which is also a codeword since the code is linear.

The weight of $\mathbf{c}_{\textnormal{new}}$ is either $m-3$ or $m-1$, depending on whether $i^* \in \{i_3,\ldots,i_m\}$ or not. Hence, $\mathbf{c}$ can be decomposed as $\mathbf{c} = \mathbf{c}_{\textnormal{new}} \oplus \mathbf{c}'$.
We now apply the same procedure iteratively to $\mathbf{c}_{\textnormal{new}}$. At each step, the weight strictly decreases, so the process terminates in finitely many steps. Since Hamming codes contain no codewords of weight $1$ or $2$, the procedure eventually expresses $\mathbf{c}$ as a sum of weight-$3$ codewords in $\mathcal{W}_3$. This completes the proof.

\bibliography{REF}
\bibliographystyle{IEEEtran}
\end{document}